\documentclass[11pt,a4paper]{article}

\usepackage{amsfonts}
\usepackage{amsmath}
\usepackage{amssymb}
\usepackage{amsthm}
\usepackage{fullpage}
\usepackage{mleftright}
\usepackage{color}
\usepackage{faktor}

\DeclareFontFamily{U}{mathx}{\hyphenchar\font45}
\DeclareFontShape{U}{mathx}{m}{n}{
	<5> <6> <7> <8> <9> <10>
	<10.95> <12> <14.4> <17.28> <20.74> <24.88>
	mathx10
}{}
\DeclareSymbolFont{mathx}{U}{mathx}{m}{n}
\DeclareFontSubstitution{U}{mathx}{m}{n}
\DeclareMathAccent{\widecheck}{0}{mathx}{"71}

\newcommand*{\norm}[1]{\left\|#1\right\|}
\newcommand*{\supp}{\mbox{supp}}
\newcommand*{\dist}{\mbox{dist}}
\newcommand*{\down}{\!\downarrow\!}

\newcommand{\secref}[1]{\S\ref{#1}}
\newcommand{\Z}{\mathbb{Z}}
\newcommand{\F}{\mathbb{F}}

\numberwithin{equation}{section}
\newtheorem{theorem}{Theorem}[section]
\newtheorem{lemma}[theorem]{Lemma}
\newtheorem{proposition}[theorem]{Proposition}
\newtheorem{corollary}[theorem]{Corollary}

\theoremstyle{definition}
\newtheorem{definitoin}[theorem]{Definition}

\mleftright
\begin{document}
	
\title{Good Distance Lattices from High Dimensional Expanders}
\author{
	Tali Kaufman\thanks{Bar-Ilan University, ISRAEL. Email: \texttt{kaufmant@mit.edu}. Research supported in part by ERC and BSF.} \and
	David Mass\thanks{Bar-Ilan University, ISRAEL. Email: \texttt{dudimass@gmail.com}.}}
\maketitle

\begin{abstract}
	We show a new framework for constructing good distance lattices from high dimensional expanders. For error-correcting codes, which have a similar flavor as lattices, there is a known framework that yields good codes from expanders. However, prior to our work, there has been no framework that yields good distance lattices directly from expanders. Interestingly, we need the notion of {\em high dimensional expansion} (and not only one dimensional expansion) for obtaining large distance lattices which are dense.

	Our construction is obtained by proving the existence of bounded degree high dimensional cosystolic expanders over any ring, and in particular over $\Z$. Previous bounded degree cosystolic expanders were known only over $\F_2$. The proof of the cosystolic expansion over any ring is composed of two main steps, each of an independent interest: We show that coboundary expansion over any ring of the links of a bounded degree complex implies that the complex is a cosystolic expander over any ring. We then prove that all the links of Ramanujan complexes (which are called spherical buildings) are coboundary expanders over any ring.
	
	We follow the strategy of~\cite{LMM16} for proving that the spherical building is a coboundary expander over any ring. Besides of generalizing their proof from $\F_2$ to any ring, we present it in a detailed way, which might serve readers with less background who wish to get into the field.
\end{abstract}
	
\section{Introduction}
We show in this work that high dimensional expanders imply lattices with good distance. There are constructions of good error-correcting codes from expanders~\cite{SS96,Spi96,LMSS01}, and since error-correcting codes and lattices are of a similar flavor, it is natural to expect that it would be possible to construct good lattices from expanders. However, prior to our work, no such construction has been known to exist. We provide a new framework for constructing lattices from high dimensional expanders, and show that a certain family of high dimensional expanders can be used in order to construct lattices with good distance.

\paragraph{Error-correcting codes.}
An error-correcting code is a subset of $n$-bit strings $\mathcal{C} \subseteq \{0,1\}^n$ called codewords. In coding theory, a good code has the following two (conflicting) properties: First, any two codewords are far from each other, i.e., many bit flips are required in order to transform one codeword into another. And second, there are many codewords, i.e., $\mathcal{C}$ is dense in $\{0,1\}^n$.

The fact that error-correcting codes and expander graphs are related is well known by now. The idea to construct codes from graphs was initiated by Gallager~\cite{Gal63} already in 1963. Gallager suggested to use a randomly chosen sparse bipartite graph, as explicit expanders did not exist at that time. Sipser and Spielman~\cite{SS96} at their celebrated result used expander graphs for explicit constructions of asymptotically good error-correcting codes, and this idea was taken further by others (for example~\cite{Spi96} and~\cite{LMSS01}).

\paragraph{Lattices.}
Given a real vector space $W$ with a basis $B = \{w_1,\dotsc,w_n\}$, the lattice $\mathcal{L} \subset W$ generated by $B$ is the subgroup of all integer linear combinations of $B$, i.e.,
$$\mathcal{L} = \left\{\sum_{i=1}^{n}a_iw_i \;\Big|\; a_i \in \mathbb{Z}, w_i \in B \right\}.$$

In a similar sense to error-correcting codes, a good lattice has the following two (conflicting) properties: First, any two points in the lattice are far from each other. And second, there are many lattice points, i.e., $\mathcal{L}$ is dense in $W$. Lattices and error-correcting codes do not only sound similar, but also have been proven to be related. See~\cite{CS13} for constructions of lattices from error-correcting codes.

In this work we initiate the study of the following question:\\[5pt]
\textbf{Question.} \emph{Is it possible to construct a good lattice directly from an expander?}\vspace{5pt}

We show that high dimensional expanders can be used in order to construct lattices with large distance which have the potential to be dense. We then show the existence of such expanders, proving the following theorem.
\begin{theorem}[Main]\label{thm:main-informal}
There exists an infinite family of high dimensional expanders which give rise to lattices with good distance.
\end{theorem}

Let us start by illustrating the strategy we use for constructing a lattices from expanders. Let $G=(V,E)$ be a graph with $k$ connected components, each contains $l$ vertices. For each connected component $S \subset V$, define its characteristic vector $\mathbf{1}_S$ which is $1$ on every vertex $v \in S$ and $0$ on every vertex $v \notin S$. We measure the size of each vector by its hamming weight, i.e., the number of entries which are not $0$. Now, consider the lattice generated by the $\mathbb{Z}$-span of these vectors. This lattice has dimension $k$ and distance $l$. Of course we have $k\cdot l \le |V|$, so we cannot hope to have both dimension and distance linear in $|V|$. Surprisingly, when moving to higher dimensions we can have both at the same time. So we are looking for higher dimensional analogs of graphs and we want that all of their (high dimensional) connected components would be large. This would give us lattices with large distance with the potential to have also large dimension.

The high dimensional analogs of graphs are called \emph{simplicial complexes}. A $d$-dimensional simplicial complex is a $(d+1)$-hypergraph with a closure property, namely, for any $(d+1)$-hyperedge in the complex, all of its subsets are also in the complex. An hyperedge is called a \emph{face} of the complex, and its dimension is one less than its cardinality. For a complex $X$, we denote by $X(0)$ the set of $0$-dimensional faces, which are the vertices, by $X(1)$ the $1$-dimensional faces, which are the edges, and so on up to $X(d)$, which are the top dimensional faces. As an example, a $1$-dimensional complex is just a graph, and a $2$-dimensional complex contains also triangles in addition to vertices and edges. Let us introduce two more definitions regarding high dimensional complexes.
\begin{enumerate}
	\item For any $0 \le k \le d-1$, the \emph{$k$-skeleton} of $X$ is the complex obtained by taking only faces of dimension $\le k$ in $X$. In particular, the $1$-skeleton of $X$ is its underlying graph (ignoring the higher dimensional faces).
	\item For any face $\sigma \in X$, the \emph{link} of $\sigma$ is the subcomplex obtained by taking all faces in $X$ which contain $\sigma$ and removing $\sigma$ from all of them, formally defined as $X_\sigma = \{\tau \setminus \sigma \;|\; \sigma \subseteq \tau \in X \}$. Note that $X_\sigma$ is a subcomplex of dimension $d-|\sigma|$.
\end{enumerate}

\subsection{Cohomology of complexes}
The high dimensional analogs of connected components are captured by the \emph{cohomology groups} of the complex. Let us consider the simple case of $d=1$, so $X=(V,E)$ is a graph. In this case, there is only one cohomology group, which corresponds to the connected components in the graph. A connected component in $X$ is a subset of vertices $S \subseteq V$ such that all edges are either inside $S$ or outside of $S$. The graph is connected if the only subsets $S$ which satisfy this criterion are trivial, i.e., $S=\emptyset$ or $S=V$. Instead of thinking of subsets of vertices, we could consider functions which give an integer value to each vertex, namely, $f:V\to\mathbb{Z}$. The equivalent way of saying that an edge $\{u,v\}$ is inside or outside $S$ is if $f(u) - f(v) = 0$. Thus, the graph is connected if the only functions for which all edges are either inside or outside of them are the constant functions. The $0$-cohomology of $X$, denoted by $H^0(X;\mathbb{Z})$, is the group of functions which vanish on all edges, where we identify functions that differ by a constant function as equivalent. If $X$ is connected, then $H^0(X;\mathbb{Z})$ is trivial, since any function that vanishes on all edges is constant and hence equivalent to the $\mathbf{0}$ function. If $X$ has more than one connected component, then its $0$-cohomology is an abelian group generated by the functions $\mathbf{1}_S$ for each connected component $S \subset V$.

Let us now move to dimension $2$, so $X = (V,E,T)$ is a $2$-dimensional simplicial complex. Now $X$ has two cohomology groups, where its $0$-cohomology is the same as before, and its $1$-cohomology corresponds to functions which vanish on all triangles. For technical reasons, we consider the faces in $X$ with orientation and say that $X$ has all possible orientations of its faces, i.e., if $(u,v) \in E$ then also $(v,u) \in E$. Consider the set of all antisymmetric functions $f:E\to \mathbb{Z}$ which assign an integer value to each edge in the complex (antisymmetric means that $f((u,v)) = -f((v,u))$ for any $(u,v) \in E$). We say that $f$ vanishes on all triangles if for any triangle $(u,v,w) \in T$, $f((u,v)) + f((v,w)) - f((u,w)) = 0$. In this case we also have functions which trivially vanish on all triangles: Take any subset of vertices $S \subseteq V$ and define the function $f$ by $f((u,v)) = \mathbf{1}_S(v) - \mathbf{1}_S(u)$. Then for any triangle $(u,v,w) \in T$,
$$f((u,v)) + f((v,w)) - f((u,w)) =
\mathbf{1}_S(v) - \mathbf{1}_S(u) + \mathbf{1}_S(w) - \mathbf{1}_S(v) - (\mathbf{1}_S(w) - \mathbf{1}_S(u)) = 0.$$
The $1$-cohomology of $X$, denoted by $H^1(X;\mathbb{Z})$, is the group of functions $f:E\to\mathbb{Z}$ which vanish on all triangles, where we identify functions that differ by a trivially vanishing function as equivalent.

In general, the $k$-cohomology captures the amount of functions $f:X(k) \to \Z$ which vanish on all $(k+1)$-dimensional faces. For any $k$ we have functions which trivially vanish on all $(k+1)$-dimensional faces, so again we consider two functions as equivalent if they differ by a trivially vanishing function.

As could be understood from the description in the above paragraphs, the cohomology groups are actually \emph{quotient spaces}. We start with the space of all functions $f:X(k) \to \Z$. Out of that we take the subspace of functions which vanish on all $(k+1)$-dimensional faces. Then we take a quotient space by identifying two functions as equivalent if they differ by a trivially vanishing function. When we construct a lattice from this quotient space, we take as a basis for the lattice a minimal representative from each equivalence class, and take the $\Z$-span of these basis elements.

In the case of graphs we could not have many connected components which are all large, so the $0$-cohomology group could not have many large elements. But for high dimensional complexes, it could be the case that for some $k > 0$, the $k$-cohomology would contain many elements, where all of them are large. Then the question we address is which complexes have only large elements in their cohomology groups. The way we answer this question is through local considerations. Roughly speaking, if every local piece of the complex is expanding, then all the elements in its cohomology groups are large. In the following we explain this criterion.

\subsection{High dimensional expanders}
Our aim in this section is to introduce briefly the notion of expansion in higher dimensions. In recent years, several definitions for high dimensional expansion have been studied. Before presenting them, let us recall expansion in graphs.

\subsubsection{Graph expansion}
\paragraph{Combinatorial expansion.}
Expander graphs have been defined explicitly by Pinsker~\cite{Pin73} as bounded degree graphs which are strongly connected. The strong connectivity of a graph is measured by its Cheeger constant, defined as follows. Let $G=(V,E)$ be a $k$-regular graph. For any subset of vertices $S \subseteq V$, denote by $E(S,\bar{S})$ the set of edges with one endpoint inside $S$ and one endpoint outside of $S$. Note that $G$ is connected if and only if $E(S,\bar{S}) \ne \emptyset$ for any $S \subseteq V$ which is not $\emptyset$ or $V$. The Cheeger constant of $G$ is defined as
$$h(G) = \min_{\emptyset \ne S \subsetneq V}\frac{|E(S,\bar{S})|}{\dist(S,\{\emptyset,V\})},$$
where $\dist(S,\{\emptyset,V\})$ is measured with hamming distance, so $\dist(S,\{\emptyset,V\}) = \min\{|S|,|V\setminus S| \}$. The graph $G$ is said to be an $\varepsilon$-combinatorial expander if $h(G) \ge \varepsilon k$ for some constant $\varepsilon > 0$.

\paragraph{Spectral expansion.}
Another notion of expansion of graphs is captured by their spectral gap. Let $A = A(G)$ be the graph's adjacency matrix, and denote by $\lambda_1 \ge \lambda_2 \ge \dotsb \ge \lambda_{|V|}$ the eigenvalues of $A$. Note that since $G$ is $k$-regular then $\lambda_1 = k$. We say that $G$ is an $\varepsilon$-spectral expander if $\lambda_2/k \le \varepsilon$ for some constant $\varepsilon > 0$. As it turns out, spectral expansion controls the graph's pseudorandom behavior. This is demonstrated by the following mixing lemma.
\begin{lemma}[Expander mixing lemma]\label{lem:expander-mixing-lemma}
Let $G=(V,E)$ be an $\varepsilon$-spectral expander. Then for any subset of vertices $S \subseteq V$,
$$\frac{|E(S)|}{|E|} \le \left(\frac{|S|}{|V|} \right)^2 + \varepsilon\frac{|S|}{|V|},$$
where $E(S)$ denotes the set of edges with both endpoints in $S$.
\end{lemma}

Note that the expectation of the fraction of edges inside $S$ in a random graph is $(|S|/|V|)^2$, so the spectral expansion measures how close is $G$ to the behavior of a random graph.

\subsubsection{High dimensional expansion}
There are several different ways to extend the notion of expansion from graphs to simplicial complexes. In the following we provide an informal definition of them just for dimension $2$, for formal definitions see~\secref{sec:preliminaries}.

\paragraph{Coboundary expansion.}
The notion of \emph{coboundary expansion} has been introduced by Linial and Meshulam~\cite{LM06} in their work on homological connectivity of random complexes, and independently by Gromov~\cite{Gro10} in his work on the topological overlapping property. Coboundary expansion is a natural extension of graph's combinatorial expansion from a homological point of view.

Let $X=(V,E,T)$ be a $2$-dimensional simplicial complex and assume that any vertex is contained in $k_1$ edges and any edge is contained in $k_2$ triangles. For any subset of vertices $S \subseteq V$, let $\delta(S) \subseteq E$ be the coboundary of $S$, defined as the set of edges which $S$ touches odd many times. Note that $\delta(S) = E(S,\bar{S})$, so this measures exactly the set of outgoing edges of $S$. Recall that for $S \in \{\emptyset,V\}$, $\delta(S)$ is trivially empty. For any subset of edges $F \subseteq E$, let $\delta(F) \subseteq T$ be the coboundary of $F$, defined as the set of triangles which $F$ touches odd many times. For subsets of edges we also have sets which their coboundary is trivially empty: Consider the set of edges $F$ between $S$ and $V\setminus S$ for some subset $S \subseteq V$. This set is called a cut in a graph. Note that for any $F$ which is a cut, $\delta(F) = \emptyset$. We say that $X$ is an $\varepsilon$-coboundary expander if:
\begin{enumerate}
	\item For any $S \subset V$, $S \notin \{\emptyset, V \}$,
	$$\frac{|\delta(S)|}{\dist(S,\{\emptyset, V \})} \ge \varepsilon k_1.$$
	\item For any $F \subset E$, $F \notin \{\mbox{cuts} \}$,
	$$\frac{|\delta(F)|}{\dist(F,\{\mbox{cuts}\})} \ge \varepsilon k_2.$$
\end{enumerate}

Condition 1 in the definition is exactly the Cheeger constant of the underlying graph of $X$, and condition 2 is its high dimensional analog for the edges of $X$.

\paragraph{Cosystolic expansion.}
Coboundary expansion is a very strong requirement and as of now it is not known if bounded degree coboundary expanders of dimension greater than $1$ even exist. A relaxation of coboundary expansion, called \emph{cosystolic expansion}, has been defined by~\cite{EK16}, and the existence of bounded degree cosystolic expanders of any dimension has been proven in~\cite{KKL14,EK16}.

In cosystolic expansion, we allow non-trivial sets to have coboundary $0$ as long as they are large. We call the sets which have coboundary $0$, the \emph{cocycles}. Then $X$ is said to be an $(\varepsilon,\mu)$-cosystolic expander if:
\begin{enumerate}
	\item For any $S \subseteq V$, $S \notin \{\emptyset, V\}$:
	\begin{enumerate}
		\item If $|\delta(S)| = 0$ then $|S| \ge \mu|V|$.
		\item Otherwise, $$\frac{|\delta(S)|}{\dist(S,\{\mbox{cocycles}\})} \ge \varepsilon k_1.$$
	\end{enumerate}
	\item For any $F \subseteq E$, $F \notin \{\mbox{cuts}\}$:
	\begin{enumerate}
		\item If $|\delta(F)| = 0$ then $|F| \ge \mu|E|$.
		\item Otherwise, $$\frac{|\delta(F)|}{\dist(F,\{\mbox{cocycles}\})} \ge \varepsilon k_2.$$
	\end{enumerate}
\end{enumerate}

Condition 1 in the definition is like saying that the underlying graph of $X$ is composed of many large connected components, where each of them is an $\varepsilon$-combinatorial expander, and condition 2 is its high dimensional analog for the edges of $X$.

\subsection{Constructing lattices from high dimensional expanders}
The idea of constructing good distance lattices from high dimensional expanders is the following. We take a complex which is a cosystolic expander, so we know that it has only large non-trivial cocycles (this is condition (a) in the definition above). We consider its cohomology group, which is a quotient space of non-trivial cocycles, where we identify two cocycles as equivalent if they differ by a trivial cocycle. Now we take a minimal representative from each equivalence class as a basis for the lattice and consider their $\Z$-span. But for that we need to know that the complex is a cosystolic expander over $\Z$. Let us explain what that means.

\subsubsection{Cosystolic expansion over $\Z$}
Note that both of the above definitions of coboundary and cosystolic expansion relate to subsets of faces. This is identical to considering functions from the vertices to $\F_2$ and from the edges to $\F_2$. These definitions extend naturally to functions over any ring (with a small modification to the coboundary operator, see~\secref{sec:preliminaries}).

In the work of~\cite{EK16}, they showed the existence of cosystolic expanders for functions over $\F_2$. This is not enough for us as we need cosystolic expansion for functions over $\Z$: If $X$ is a cosystolic expander with respect to functions over $\Z$, then any element in the $\Z$-span of non-trivial cocycles is large (this is part of the definition of cosystolic expansion). Therefore, for our lattice construction we need to prove the existence of cosystolic expanders over $\Z$.

We generalize the proof of~\cite{EK16} so it would work over any ring. First, we have translated their proof to language of probabilities, which makes the proof simpler even though the main ideas remain the same. Second, when working over general rings and not only over $\F_2$, there is the matter of orientations of faces which is needed to be taken care of. Previous works did not worry about orientations as they worked only over $\F_2$, where addition and subtraction are the same. We work over general rings, hence we have to cope with orientations of faces. This was not done in previous works.

The key point for proving cosystolic expansion over any ring is to show that any cochain can be decomposed into local parts, so its global expansion would be implied by the expansion of its local parts. Our main technical contribution is the following theorem.

\begin{theorem}[Existence of good dimension, informal, for formal see~\ref{thm:existence-of-good-dimension}]\label{thm:existence-of-good-dimension-informal}
If the underlying graph of any link in $X$ is a good enough spectral expander, then for any function $f:X(k) \to R$ for any ring $R$, there exists a dimension $0 \le i \le k$, such that $f$ can be decomposed to local parts of $i$-dimensional faces and most of the expansion of $f$ is implied by local expansion in the links of $i$-dimensional faces.
\end{theorem}

The above theorem tells us that the global expansion of a complex can be deduced from the expansion of its links. Thus, in order to get cosystolic expansion over any ring we only need to show that the links are good, i.e., their underlying graph is a spectral expander and they are coboundary expanders over any ring. We show that Ramanujan complexes have this property.

\subsubsection{Ramanujan complexes and their links}
Ramanujan complexes are the high dimensional analogs of the celebrated LPS graphs~\cite{LPS88}. LPS graphs are constructed by taking quotients of the infinite tree, which is the best expander possible. The infinite tree has an high dimensional analog, called the Bruhat-Tits building. This led~\cite{LSV05.1} to study quotients of it as a generalization of LPS graphs. By taking quotients of the Bruhat-Tits building,~\cite{LSV05.2} achieve an explicit construction of bounded degree simplicial complexes which locally look like the infinite object. These complexes are called \emph{Ramanujan complexes}. (For more on Ramanujan complexes see~\cite{Lub14}.)

Every link of a Ramanujan complex is a very symmetric complex called the \emph{spherical building} (more details on the spherical building are presented in \secref{sec:spherical-building}). The spherical building by itself is of unbounded degree, since the number of faces incident to any vertex grows with the number of vertices in the complex. But as links of a Ramanujan complex, the global complex is of a bounded degree. In~\cite{EK16}, the authors showed that the $1$-skeleton of the spherical building is an excellent spectral expander (its expansion quality is controlled by a parameter called the \emph{thickness} of the building), so it is left to show that the spherical building is a coboundary expander over any ring.

In~\cite{LMM16}, the authors showed that the spherical building is a coboundary expander over $\F_2$. This is not enough for us as we need expansion over $\Z$. We generalize the work of~\cite{LMM16} by taking care of orientations of faces (which was not necessary in their work since they proved only for $\F_2$). We show that with some modifications, which take orientations into account, the proof of~\cite{LMM16} can work over any ring. We prove the following theorem.

\begin{theorem}[The spherical building is a coboundary expander]
The spherical building is a coboundary expander over any ring.
\end{theorem}

Since we got that the links of (thick enough) Ramanujan complexes are spectral and coboundary expanders over any ring, we achieve the following theorem.

\begin{theorem}[Ramanujan complexes are cosystolic expanders over any ring]\label{thm:ramanujan-complexes-are-cosystolic-expanders}
For a thick enough $d$-dimensional Ramanujan complex, its $(d-1)$-skeleton is a cosystolic expander over any ring.
\end{theorem}

\subsubsection{The dimension of the lattice}

Up to now we got that Ramanujan complexes are cosystolic expanders over $\Z$, and thus can be used in order to construct lattices with large distance. It is left to consider the dimension of these lattices. It is clear that the dimension of the lattice is controlled by the amount of elements in the cohomology groups over $\Z$. However, usually it is easier to understand the cohomology groups over $\F_2$, where understanding them over $\Z$ could be a very hard task. Moreover, as proven in~\cite{KKL14}, we already know that Ramanujan complexes have non-trivial cohomology groups over $\F_2$ (actually we know that the number of elements in $H^1(X;\F_2)$ is logarithmic in the size of the complex~\cite{Lub}). Luckily, there is a way to relate cohomology groups over $\Z$ to cohomology groups with other coefficients. In a way, the cohomology groups over $\Z$ are considered universal, so they determine the cohomology groups with any other coefficients. This is done by \emph{the universal coefficient theorem}, which we introduce next.

\paragraph{The universal coefficient theorem.}
Since understanding this theorem requires a lot of background in algebraic topology, we introduce it in some sort of informal way. Any finitely generated abelian group $H$ has a decomposition to its free part and torsion part. The torsion part contains all elements of finite order (i.e., all $h \in H$ for which there exists $n \in \mathbb{N}$ such that $nh = 0$). Thus $H$ can be decomposed to
$$H \cong \Z^k \oplus T(H),$$
where $k$ is the number of free generators in $H$, and $T(H)$ is its torsion subgroup. 
The universal coefficient theorem gives us information about this decomposition for the cohomology groups. In particular, it tells us that there exist $k,l \ge 0$ and $n_1, n_2, \dotsc, n_l \ge 0$, such that:
\begin{enumerate}
	\item $H^1(X;\F_2) \cong \Z^k \oplus \bigoplus_{i=1}^l \F_{2^{n_i}}$.
	\item $H^1(X;\Z) \cong \Z^k$.
	\item $H^2(X;\Z)$ contains $\bigoplus_{i=1}^l \F_{2^{n_i}}$ in its decomposition.
\end{enumerate}

Therefore we have the following corollary.

\begin{corollary}
If $H^1(X;\F_2)$ is large then either $H^1(X;\Z)$ is large or $H^2(X;\Z)$ is large.
\end{corollary}

A good situation for us would be that the dominant part of $H^1(X;\F_2)$ comes from the free part of the decomposition and not from the torsion part. This would imply that the lattice we construct has a large dimension, as it has many free generators. However, it is still an open question whether this is the case in Ramanujan complexes, and this theorem gives us a lead how to approach the study of the cohomology groups over $\Z$.

\subsection{The lattice construction}

Overall we have that the links of Ramanujan complexes are good spectral and coboundary expanders over $\Z$, and by theorem~\ref{thm:existence-of-good-dimension-informal}, this implies that Ramanujan complexes are cosystolic expanders over $\Z$. So taking the $\Z$-span of minimal representatives of their cohomology groups yields lattices with good distance.

The advantage of this construction is that now we can have both distance and dimension large. Recall that in graphs there is a trivial tradeoff, as the multiplication of the distance and the dimension of the lattice is bounded by $n$, where $n$ is the number of vertices in the graph. When using this framework for high dimensional expanders, it could be the case that both distance and dimension are large. In particular, we show that the distance of lattices constructed from Ramanujan complexes is linear in $n$, and by the universal coefficient theorem, it might be that their dimension is logarithmic in $n$. This yields that their multiplication could potentially be of order $n \log n$. In general we expect that there should be bounded degree co-systolic expanders over some rings that would be resulted in good distance lattices which are sufficiently dense.

\subsection{Discussion and future work}
Recent works showed that cosystolic expansion over $\mathbb{F}_2$ implies Gromov's topological overlapping property~\cite{EK16}. We show here that cosystolic expansion over $\Z$ implies the construction of lattices with good distance. We believe that cosystolic expansion over general rings should have far reaching applications that are beyond one's imagination. For instance, it could lead to new constructions of locally testable codes.

\subsection{Organization}
We start with a preliminaries section that contains the basics of cochains with norms in high dimensional complexes. In section~\ref{sec:cosystolic-expansion} we prove theorem~\ref{thm:existence-of-good-dimension-informal} and show how it implies cosystolic expansion over any ring. In section~\ref{sec:spherical-building} we introduce the links of Ramanujan complexes, which are called spherical buildings, and prove that they are coboundary expanders over any ring.

\section{Preliminaries}\label{sec:preliminaries}
Let $X$ be a $d$-dimensional simplicial complex. For any $-1 \le k \le d$, denote by $X(k)$ the set of $k$-dimensional faces of $X$ (where $X(-1) = \{\emptyset\}$ contains the only $-1$-dimensional face, which is the face with $0$ vertices). An ordered set $\vec{\sigma} = (v_0,v_1, \dotsc,v_k)$ is an \emph{ordered face} of $X$ if the unordered set $\sigma = \{v_0,v_1,\dotsc,v_k \}$ is a face of $X$. Denote by $\vec{X}(k)$ the set of ordered $k$-dimensional faces of $X$. The space of \emph{$k$-cochains} over a ring $R$ is defined as
$$C^k = C^k(X;R) = \{f : \vec{X}(k) \to R \;|\; f \mbox{ is antisymmetric} \},$$
where $f$ is antisymmetric if for any permutation $\pi \in Sym(k+1)$,
$$f((v_{\pi(0)},v_{\pi(1)},\dotsc,v_{\pi(k)})) = sgn(\pi)f((v_0,v_1,\dotsc,v_k)).$$

Note that for $R = \mathbb{F}_2$, the $k$-cochains are just subsets of $X(k)$ (where we identify a subset of faces with its characteristic function). In the works of~\cite{LMM16} and~\cite{EK16}, which we generalize in this paper, the authors worked only with cochains over $\mathbb{F}_2$ so they did not have to worry about ordered faces and change of signs. We let the cochains to be over any ring so we need to take these considerations into account.

We measure the size of a cochain according to its hamming weight with proportion to the top dimension of the complex, as follows. Let $r_d,r_{d-1},\dotsc,r_{-1}$ be a sequence of random faces of $X$, where $r_d$ is distributed uniformly on $X(d)$, and for any $k<d$, $r_k$ is obtained by removing a uniformly random vertex from $r_{k+1}$. All the probabilities we measure in this work would be over this distribution of random faces. For any $k$-cochain $f \in C^k$, we denote its support by $A = \supp(f) = \{\sigma \in X(k) \;|\; f(\sigma) \ne 0\}$, and define its norm to be $\norm{f} = \norm{A} = \Pr[r_k \in A]$. (Note that the support of $f$ is a set of unordered faces, and it is well defined even though the cochain is defined on ordered faces, since it does not matter which ordering we take.)

For any $\vec{\sigma} = (v_0, v_1, \dotsc, v_k)$ we denote by $\vec{\sigma}\setminus \{v_i\} = (v_0,\dotsc, v_{i-1},v_{i+1},\dotsc,v_k)$ the ordered $(k-1)$-face obtained by removing $v_i$ from $\sigma$. The $k$-coboundary operator $\delta = \delta^k: C^k \to C^{k+1}$ is defined as
$$\delta(f)(\vec{\sigma}) = \sum_{i=0}^{k+1}(-1)^i f(\vec{\sigma}\setminus \{v_i\}).$$

Denote by $B^k = \mbox{Im}(\delta^{k-1}) = \{\delta^{k-1}(f) \;|\; f \in C^{k-1}\}$ the $k$-coboundaries of $X$, and by $Z^k = \ker(\delta^k) = \{f \in C^k \;|\; \delta^k(f) = 0\}$ the $k$-cocycles of $X$. The $k$-cohomology group is the quotient space $H^k = Z^k/B^k$. The distance of a $k$-cochain $f \in C^k$ from the $k$-coboundaries is defined as $\dist(f,B^k) = \min\{\norm{f - b} \;|\; b \in B^k \}$. Similarly, the distance from the $k$-cocycles is defined as $\dist(f,Z^k) = \min\{\norm{f - z} \;|\; z \in Z^k \}$.

We can now present the notion of coboundary expansion as was introduced by Linial-Meshulam~\cite{LM06} and by Gromov~\cite{Gro10}.

\begin{definitoin}[Coboundary expansion]\label{def:coboundary-expansion}
	Let $X$ be a $d$-dimensional simplicial complex and $R$ a ring. $X$ is called an \emph{$\varepsilon$-coboundary expander} over $R$, if for any $k$-cochain which is not a $k$-coboundary $f \in C^k(X;R) \setminus B^k(X;R)$, $0 \le k \le d-1$,
	$$\frac{\norm{\delta(f)}}{\dist(f, B^k(X;R))} \ge \varepsilon.$$
\end{definitoin}

As it turns out, coboundary expansion is a very strong requirement. Currently it is not known whether bounded degree coboundary expanders of dimension $\ge 2$ even exist. This leads to the relaxation of coboundary expansion, called cosystolic expansion, which was introduced by~\cite{EK16}, and is defined as follows.

\begin{definitoin}[Cosystolic expansion]\label{def:cosystolic-expansion}
	Let $X$ be a $d$-dimensional simplicial complex and $R$ a ring. $X$ is called an \emph{$(\varepsilon, \mu)$-cosystolic expander} over $R$, if:
	\begin{enumerate}
		\item For any $f \in C^k(X;R) \setminus Z^k(X;R)$, $0 \le k \le d-1$,
		$$\frac{\norm{\delta(f)}}{\dist(f, Z^k(X;R))} \ge \varepsilon.$$
		\item For any $z \in Z^k(X;R) \setminus B^k(X;R)$, $0 \le k \le d-1$,
		$$\norm{z} \ge \mu.$$
	\end{enumerate}
\end{definitoin}

Recall that for any $\sigma \in X$, its \emph{link} is the subcomplex obtained by taking all the faces in $X$ which contain $\sigma$, and removing $\sigma$ from all of them. Since the link of $\sigma$ is a complex by itself, we can talk about cochains and norms in the link. Consider a $(k-|\sigma|)$-cochain in the link of $\sigma$, $f\in C^{k-|\sigma|}(X_\sigma;R)$. Its norm in the link is the probability that a random face would fall in $\supp(f)$ when the top face is distributed uniformly over the top faces in $X_\sigma$. Thus,
$$\norm{f} = \Pr[r_k \setminus \sigma \in \supp(f) \;|\; r_{|\sigma|-1} = \sigma],$$
where $\norm{f}$ is the norm in the link, and $r_k,r_{|\sigma|-1}$ are the random faces chosen in $X$.

From now on we fix for any face in the complex an arbitrary choice of ordering, so for any $\sigma \in X$ there is one fixed ordered face $\vec{\sigma}$ which corresponds to it. The choice of ordering does not matter, it just has to be consistent. For any $k$-cochain $f \in C^k$ and any face $\sigma \in X$, we define the \emph{localization} of $f$ to the link of $\sigma$, denoted by $f_\sigma$, as follows. For any ordered $(k-|\sigma|)$-face $\vec{\tau} \in \vec{X}_\sigma(k-|\sigma|)$, we define $f_\sigma(\vec{\tau})=f(\vec{\sigma\tau})$, where $\vec{\sigma\tau} \in \vec{X}(k)$ is the ordered $k$-face obtained by concatenating $\vec{\tau}$ to $\vec{\sigma}$. We say that a cochain $f \in C^k$ is \emph{minimal} if $\norm{f} = \dist(f, B^k)$. We say that $f$ is \emph{locally minimal} if its localization to any link is minimal, i.e., if $f_\sigma$ is minimal in $X_\sigma$ for any $\emptyset \ne \sigma \in X$.

The following two lemmas regarding minimal and locally minimal cochains will be necessary later.

\begin{lemma}[Minimal cochains are closed under inclusion]\label{lem:minimal-cochains-are-closed-under-inclusion}
Let $X$ be a $d$-dimensional simplicial complex and $R$ a ring. For any $f,g\in C^k(X;R)$, $0 \le k \le d$, if $f$ is a minimal cochain and $g(\vec{\sigma}) = f(\vec{\sigma})$ for any $\sigma \in \supp(g)$, then $g$ is a minimal cochain.
\end{lemma}
\begin{proof}
Note that for any $k$-cochain $h\in C^k$,
\begin{equation}\label{eq:subset-of-minimal-cochain-is-minimal-1}
\norm{f-h} - \norm{g-h} \le \norm{f-g} = \norm{f} - \norm{g},
\end{equation}
where the equality follows by the fact that $g(\vec{\sigma}) = f(\vec{\sigma})$ for any $\sigma \in \supp(g)$. Then for any $k$-coboundary $b\in B^k$,
$$\norm{g} = \norm{g} + \norm{f} - \norm{f} \le
\norm{g} + \norm{f-b} - \norm{f} \le \norm{g-b},$$
where the first inequality follows by the fact that $f$ is a minimal cochain, and the second inequality follows by~\eqref{eq:subset-of-minimal-cochain-is-minimal-1}.
\end{proof}

\begin{lemma}[Minimal cochain is also locally minimal]\label{lem:minimal-cochain-is-also-locally-minimal}
Let $X$ be a $d$-dimensional simplicial complex and $R$ a ring. For any $f \in C^k(X;R)$, $0 \le k \le d$, if $f$ is minimal, then $f$ is also locally minimal.
\end{lemma}
\begin{proof}
Let $f \in C^k(X;R)$ be a minimal cochain. Assume towards contradiction that $f$ is not locally minimal. There exists a face $\emptyset \ne \sigma \in X$ and a cochain $h \in C^{k-|\sigma|-1}(X_\sigma;R)$ such that
\begin{equation}\label{eq:minimal-cochain-is-also-locally-minimal-1}
\norm{f_\sigma - \delta(h)} < \norm{f_\sigma}.
\end{equation}

Define $g \in C^{k-1}(X;R)$ by $g(\sigma\tau) = h(\tau)$ for any $\tau \in X_\sigma$, and for any other face $g(\tau) = 0$. Note that $g_\sigma = h$, then by~\eqref{eq:minimal-cochain-is-also-locally-minimal-1},
$$\norm{f - \delta(g)} < \norm{f},$$
in contradiction to the minimality of $f$. It follows that $f$ is locally minimal.
\end{proof}

We have one more definition we want to present in this section.
\begin{definitoin}[Skeleton expansion]\label{def:skeleton-expansion}
	Let $X$ be a $d$-dimensional simplicial complex. $X$ is called an \emph{$\alpha$-skeleton expander}, if for any subset of vertices $S \subseteq X(0)$,
	$$\norm{E(S)} \le \norm{S}^2 + \alpha\norm{S},$$
	where $E(S)$ denotes the set of edges with both endpoints in $S$.
\end{definitoin}

\section{Cosystolic expansion}\label{sec:cosystolic-expansion}
Our aim in this section is to show that good links imply cosystolic expansion over any ring. In~\cite{KKL14}, the authors showed that for bounded degree complexes, cosystolic expansion of the $(d-1)$-skeleton is implied by the expansion of small cochains which are locally minimal. Let us define this sort of small-set expansion.

\begin{definitoin}[Small-set expander]
Let $X$ be a $d$-dimensional simplicial complex and $R$ a ring. $X$ is called an \emph{$(\varepsilon,\mu)$-small-set expander} over $R$, if for any $f \in C^k(X;R)$, $0 \le k \le d-1$,
$$f \mbox{ is locally minimal and } \norm{f} \le \mu \quad\Rightarrow\quad \norm{\delta(f)} \ge \varepsilon\norm{f}.$$
\end{definitoin}

We start by showing that small-set expansion implies cosystolic expansion, and then we will show that good links imply small-set expansion.

\subsection{Small-set expansion implies cosystolic expansion}
The criterion of having only large non-trivial cocycles is immediate from the small-set expansion, and this actually holds for unbounded degree complexes as well.

\begin{proposition}[Small-set expansion implies large non-trivial cocycles]\label{pro:small-set-expansion-implies-large-non-trivial-cocycles}
Let $X$ be a $d$-dimensi-onal $(\varepsilon,\mu)$-small-set expander over a ring $R$. For any $z \in Z^k(X;R) \setminus B^k(X;R)$, $0 \le k \le d-1$, it holds that $\norm{z} \ge \mu$.
\end{proposition}
\begin{proof}
Note that it is enough to show that it holds for minimal non-trivial cocycles, since if $z \in Z^k(X;R) \setminus B^k(X;R)$ is not minimal, then there exists a coboundary $b \in B^k(X;R)$ such that $\norm{z} \ge \norm{z - b}$ and $z - b \in Z^k(X;R) \setminus B^k(X;R)$ is minimal.

Let $z \in Z^k(X;R)$ be a minimal cocycle. We show that if $\norm{z} < \mu$, then $z \in B^k(X;R)$. Since $z$ is minimal, by lemma~\ref{lem:minimal-cochain-is-also-locally-minimal} $z$ is also locally minimal, so by the small-set expansion, $\norm{\delta(z)} \ge \varepsilon\norm{z}$. But on the other hand $z \in Z^k(X;R)$ so $\norm{\delta(z)} = 0$. It follows that $\norm{z} = 0$, so $z \in B^k(X;R)$ as required.
\end{proof}

We say that a complex $X$ is \emph{$Q$-bounded degree} if for any $v \in X(0)$, $|X_v| \le Q$. In order to show that small-set expansion implies that any cochain in a bounded degree complex expands with respect to the cocycles (first criterion in the definition of cosystolic expansion), we need the following lemma.

\begin{lemma}\label{lem:fix-cochain-to-locally-minimal-is-small}
Let $X$ be a $d$-dimensional $Q$-bounded degree simplicial complex and $R$ a ring. For any $f \in C^k(X;R)$, $0 \le k \le d-1$, there exists $g \in C^{k-1}(X;R)$ such that:
\begin{enumerate}
	\item $\norm{g} \le Q^2\norm{f}$.
	\item $f - \delta(g)$ is locally minimal.
	\item $\norm{f - \delta(g)} \le \norm{f}$.
\end{enumerate}
\end{lemma}
\begin{proof}
We prove by induction on $\norm{f}$. For the base case, $\norm{f} = 0$, the claim holds trivially for $g = 0$. Assume the claim holds for any cochain $f'$ with $\norm{f'} < \norm{f}$. Now, if $f$ is locally minimal, the claim holds for $g = 0$. Otherwise, there exists some $\sigma \in X$ such that $f_\sigma$ is not minimal in $X_\sigma$. So there exists a cochain in the link of $\sigma$, $h \in C^{k-1-|\sigma|}(X_\sigma;R)$, such that $\norm{f_\sigma + \delta(h)} < \norm{f_\sigma}$. Define $g' \in C^{k-1}(X;R)$ by $g'(\sigma\tau) = h(\tau)$ for any $\tau \in X_\sigma$ and for any other face $g'(\tau) = 0$. It follows that $\norm{f + \delta(g')} < \norm{f}$. By the induction assumption, there exists $g'' \in C^{k-1}(X;R)$ such that:
\begin{enumerate}
	\item $\norm{g''} \le Q^2\norm{f + \delta(g')}$.
	\item $f + \delta(g') + \delta(g'') = f + \delta(g' + g'')$ is locally minimal.
	\item $\norm{f + \delta(g') + \delta(g'')} = \norm{f + \delta(g' + g'')} \le \norm{f + \delta(g')} < \norm{f}$.
\end{enumerate}
Denote by $g = g' + g''$, and note that conditions 2 and 3 are satisfied. As for condition 1, note that since
$$\norm{f} =
\sum_{\sigma \in \text{supp}(f)}\Pr[r_k = \sigma] =
\sum_{\sigma \in \text{supp}(f)}\sum_{\substack{\tau \in X(d)\\\tau \supset \sigma}}\Pr[r_d = \tau \wedge r_k = \sigma] =
\sum_{\sigma \in \text{supp}(f)}\sum_{\substack{\tau \in X(d)\\\tau \supset \sigma}}\frac{1}{|X(d)|\binom{d+1}{k+1}},$$
Then
\begin{equation}\label{eq:fix-cochain-to-locally-minimal-is-small-1}
\norm{f + \delta(g')} \le \norm{f} - \frac{1}{|X(d)|\binom{d+1}{k+1}},
\end{equation}
And also, since $\supp(g')$ contains only faces which contain $\sigma$, then
\begin{equation}\label{eq:fix-cochain-to-locally-minimal-is-small-2}
\begin{aligned}
\norm{g'} &\le
\sum_{\substack{\tau \in X(k) \\ \tau \supset \sigma}}\Pr[r_k = \tau] =
\sum_{\substack{\tau \in X(k) \\ \tau \supset \sigma}}\sum_{\substack{\rho \in X(d) \\ \rho \supset \tau}}\Pr[r_d = \rho \wedge r_k = \tau] \\&=
\sum_{\substack{\tau \in X(k) \\ \tau \supset \sigma}}\sum_{\substack{\rho \in X(d) \\ \rho \supset \tau}}\frac{1}{|X(d)|\binom{d+1}{k+1}} \le
\frac{Q^2}{|X(d)|\binom{d+1}{k+1}}.
\end{aligned}
\end{equation}

Combining~\eqref{eq:fix-cochain-to-locally-minimal-is-small-1} and~\eqref{eq:fix-cochain-to-locally-minimal-is-small-2} yields
$$\norm{g} \le
\norm{g'} + \norm{g''} \le
\frac{Q^2}{|X(d)|\binom{d+1}{k+1}} +
Q^2\left(\norm{f} - \frac{1}{|X(d)|\binom{d+1}{k+1}}\right) =
Q^2\norm{f},$$
and condition 1 is satisfied as well.
\end{proof}

Now we can show the expansion criterion of any cochain (up to one dimension less) with respect to the cocycles.

\begin{proposition}[Small-set expansion implies cocycle expansion for one dimension less]\label{pro:small-set-expansion-implies-cocyle-expansion-for-one-dimension-less}
Let $X$ be a $d$-dimensional $Q$-bounded degree $(\varepsilon,\mu)$-small-set expander over a ring $R$. For any $f \in C^k(X;R) \setminus Z^k(X;R)$, $0 \le k \le d-2$, it holds that $$\frac{\norm{\delta(f)}}{\dist(f, Z^k(X;R))} \ge \min\{\mu,Q^{-2}\}.$$
\end{proposition}
\begin{proof}
Let $f \in C^k(X;R) \setminus Z^k(X;R)$, $0 \le k \le d-2$. If $\norm{\delta(f)} \ge \mu$ we are done, so assume that $\norm{\delta(f)} < \mu$. Let $g \in C^k(X;R)$ be the $k$-cochain promised by lemma~\ref{lem:fix-cochain-to-locally-minimal-is-small} when applied on $\delta(f)$. By properties 2 and 3 of lemma~\ref{lem:fix-cochain-to-locally-minimal-is-small}, $\delta(f) - \delta(g) = \delta(f - g)$ is a $(k+1)$-cochain which is locally minimal and $\norm{\delta(f - g)} \le \mu$. By the small-set expansion,
$$0 = \norm{\delta(\delta(f - g))} \ge \varepsilon\norm{\delta(f - g)},$$
so $\delta(f - g) = 0$, which means that $f - g \in Z^k(X;R)$. By property 3 of lemma~\ref{lem:fix-cochain-to-locally-minimal-is-small} we know that $\norm{g} \le Q^2\norm{\delta(f)}$, which yields
$$\norm{\delta(f)} \ge
Q^{-2}\norm{g} =
Q^{-2}\norm{f - (f - g)} \ge
Q^{-2}\cdot \dist(f, Z^k(X;R)).$$
\end{proof}

The following is an immediate corollary of propositions~\ref{pro:small-set-expansion-implies-large-non-trivial-cocycles} and~\ref{pro:small-set-expansion-implies-cocyle-expansion-for-one-dimension-less}.

\begin{corollary}[Small-set expansion implies cosystolic expansion for one dimension less]
If $X$ is a $d$-dimensional $Q$-bounded degree $(\varepsilon,\mu)$-small-set expander over a ring $R$, then the $(d-1)$-skeleton of $X$ is a $(\min\{\mu,Q^{-2}\}, \mu)$-cosystolic expander over $R$.
\end{corollary}

\subsection{Good links imply small-set expansion}

We follow the strategy of~\cite{EK16} in order to show that good links imply small-set expansion. The following theorem is what we observe as the key point of the proof. It shows that if all the links are skeleton expanders, then for small cochains there exists a dimension which most of the global expansion is determined by local expansion in the links of that dimension.

\begin{theorem}[Existence of good dimension]\label{thm:existence-of-good-dimension}
Let $X$ be a $d$-dimensional simplicial complex, $R$ a ring and $0 \le k \le d-1$. If for any $\sigma \in X$, the link $X_\sigma$ is an $\alpha$-skeleton expander, then for any constants $0 = c_{-1} \le c_0 \le \dotsb c_k \le 1$ and $k$-cochain $f \in C^k(X;R)$, if $f$ is locally minimal and $\norm{f} \le \alpha$, then there exists $0 \le i \le k$ such that
$$\norm{\delta(f)} \ge
\left(\beta_ic_i - (k+1-i)(i+1)c_{i-1} - \alpha^{2^{-d}}(k+1)(k+2)2^{k+2}\right)\norm{f},$$
where
$$\beta_i = \min\left\{\frac{\norm{\delta(g)}}{\dist(g, B^{k-|\sigma|}(X_\sigma;R))} : \sigma \in X(i),\; g \in C^{k-|\sigma|}(X_\sigma;R) \setminus B^{k-|\sigma|}(X_\sigma;R) \right\}.$$
\end{theorem}

This theorem is enough to prove that good links imply small-set expansion as follows.

\begin{theorem}[Good links imply small-set expansion]\label{thm:good-links-imply-small-set-expansion}
Let $X$ be a $d$-dimensional simplicial complex, $R$ a ring and $\beta > 0$. There exist $\varepsilon = \varepsilon(d,\beta)$ and $\alpha = \alpha(d,\beta)$, such that if for any $\sigma \in X$ the link $X_\sigma$ is an $\alpha$-skeleton expander and for any $\emptyset \ne \sigma \in X$ the link $X_\sigma$ is a $\beta$-coboundary expander over $R$, then $X$ is an $(\varepsilon,\alpha)$-small-set expander over $R$.
\end{theorem}

\begin{proof}
Let $0 < \rho < 1$,
$$
\varepsilon = (1-\rho)\left(1 + \frac{d(d-1)^{2(d-1)}}{\beta^{d-1}(1-\beta)} \right)^{\!-1},\quad\quad
\alpha = \left(\frac{\rho}{1-\rho}\cdot\frac{\varepsilon}{d(d+1)2^{d+1}}\right)^{2^d},$$
and define the following constants:
\begin{align*}
&\bullet\quad c_{-1} = 0, \\[9pt]
&\bullet\quad c_0 = \frac{\varepsilon}{(1-\rho)\beta},\\
&\bullet\quad c_i = c_0 + \frac{k^2}{\beta}c_{i-1}\quad\quad \forall i \in \{1,\dotsc,k-1\},\\[8pt]
&\bullet\quad c_k = \beta c_0 + (k+1)c_{k-1}.
\end{align*}

Note that
$$c_k = c_0\left(\beta + (k+1)\sum_{i=0}^{k-1}\left(\frac{k^2}{\beta}\right)^{\!i}\right) \le
c_0\left(\beta + \frac{(k+1)k^{2k}}{\beta^{k-1}(k^2-\beta)}\right) =
\frac{\varepsilon}{1-\rho} \left(1+\frac{(k+1)k^{2k}}{\beta^k(k^2-\beta)}\right) \le 1,
$$
so the conditions of theorem~\ref{thm:existence-of-good-dimension} are satisfied.
Let $f \in C^k(X;R)$ be a locally minimal $k$-cochain with $\norm{f} \le \alpha$, and let $0 \le i \le k$ be the good dimension promised by theorem~\ref{thm:existence-of-good-dimension}.
\begin{enumerate}
	\item If $i=k$, note that for any $\sigma \in \supp(f)$, $\norm{\delta(f_\sigma)} \ge \norm{f_\sigma}$, so theorem~\ref{thm:existence-of-good-dimension} yields
	$$\norm{\delta(f)} \ge \left(c_k - (k+1)c_{k-1} - \frac{\rho}{1-\rho}\varepsilon\right)\norm{f} \ge \varepsilon\norm{f}.$$
	\item Otherwise, by the $\beta$-coboundary expansion of the links, theorem~\ref{thm:existence-of-good-dimension} yields
	$$\norm{\delta(f)} \ge \left(\beta c_i - k^2c_{i-1} - \frac{\rho}{1-\rho}\varepsilon\right) \norm{f} \ge \varepsilon\norm{f}.$$
\end{enumerate}
\end{proof}

The rest of this section is dedicated to proving theorem~\ref{thm:existence-of-good-dimension}. We need to show that any small cochain can be decomposed into local parts such that the expansion of the local parts would imply the global expansion. In the following lemma we show that whenever all the information of a cochain is seen in a link, then its local coboundaries coincide with global coboundaries.

\begin{lemma}[Local-to-global coboundaries]\label{lem:when-link-sees-everything-coboundaries-are-global}
Let $X$ be a $d$-dimensional simplicial complex, $R$ a ring, $f\in C^k(X;R)$, $0 \le k \le d-1$, and $\sigma \in X(i)$, $i<k$. For any $\vec{\tau} \in \vec{X}_\sigma(k-i)$, if $\sigma\cup\tau\setminus\{v\} \notin \supp(f)$ for any $v\in\sigma$ and $\tau \in \supp(\delta(f_\sigma))$, then $\sigma\cup\tau\in \supp(\delta(f))$.
\end{lemma}
\begin{proof}
Let us denote $\vec{\sigma}=(v_0,\dotsc,v_i)$ and $\vec{\tau}=(v_{i+1},\dotsc,v_{k+1})$ (where $\vec{\sigma}$ is the fixed ordered face corresponding to $\sigma$). Then
\begin{align*}
\delta(f)(\vec{\sigma\tau}) &=
\sum_{j=0}^{k+1}(-1)^jf(\vec{\sigma\tau}\setminus\{v_j\}) = \sum_{j=i+1}^{k+1}(-1)^jf(\vec{\sigma\tau}\setminus\{v_j\}) \\&= (-1)^{i+1}\sum_{j=0}^{k-i}(-1)^jf_\sigma(\vec{\tau}\setminus\{v_{j+i+1}\}) =
(-1)^{i+1}\delta(f_\sigma)(\vec{\tau}) \ne 0.
\end{align*}
\end{proof}

We now define a machinery of fat faces, which essentially lets us move calculations down the dimensions. Let $\eta>0$ be a fatness constant. For any subset of $k$-faces $A \subseteq X(k)$ we define the sets of \emph{fat faces} as follows. The set of fat $k$-faces is defined as $A_k = A$, and for any $-1 \le i \le k-1$ we define the set of fat $i$-faces $A_i \subseteq X(i)$ by $$A_i = \{\sigma \in X(i) \;|\; \Pr[r_{i+1} \in A_{i+1} \;|\; r_i = \sigma] \ge \eta^{2^{k-i-1}} \}.$$

The following lemma shows that for any $-1 \le i \le k-1$, the size of $A_i$ cannot be much larger than the size of $A$.

\begin{lemma}\label{lem:upper-bound-on-fat-faces}
	Let $X$ be a $d$-dimensional simplicial complex and $\eta>0$ a fatness constant. For any subset of $k$-faces $A \subseteq X(k)$, $0 \le k \le d-1$, and $-1 \le i \le k-1$,
	$$\norm{A_i} \le \eta^{1-2^{k-i}}\norm{A}.$$
\end{lemma}
\begin{proof}
	By laws of probability, for any $-1\le j \le k-1$,
	\begin{equation}\label{eq:upper-bound-on-fat-faces-1}
	\Pr[r_j \in A_j] =
	\frac{\Pr[r_{j+1} \in A_{j+1} \wedge r_j \in A_j]}{\Pr[r_{j+1} \in A_{j+1} \;|\; r_j \in A_j]} \le \eta^{-2^{k-j-1}}\Pr[r_{j+1} \in A_{j+1}].
	\end{equation}
	Applying \eqref{eq:upper-bound-on-fat-faces-1} iteratively for $j = i,i+1,\dotsc,k-1$ finishes the proof.
\end{proof}

For any $\sigma \in X(i)$, $-1 \le i \le k$, we denote by $A\down \sigma\subseteq A$ the set of faces in $A$ which have a sequence of containments (in~\cite{EK16} it is called a ladder) of fat faces down to $\sigma$, formally,
$$A\down \sigma = \{\tau \in A \;|\; \exists \tau_{k-1}\in A_{k-1},\dotsc,\tau_{i+1}\in A_{i+1} \mbox{ s.t. }\tau \supset \tau_{k-1} \supset \dotsb \supset \tau_{i+1} \supset \sigma \}.$$

Recall that for a $k$-cochain $f \in C^k$, we denote its support by $A = \supp(f)$. So we also define $f\down\sigma$ to be the restriction of $f$ to $A\down\sigma$, formally,
$$
(f\down\sigma)(\vec{\tau}) =
\begin{cases}
f(\vec{\tau}) & \tau \in A\down \sigma, \\
0	& \mbox{otherwise}.
\end{cases}
$$

A good situation for us is that for any two fat faces which intersect on a codimension $1$ face, their intersection is a fat face. This essentially allows us to move calculations down the dimensions. We denote by $\Upsilon \subseteq X(k+1)$ the set of bad $(k+1)$-faces, for which a bad situation exists, formally,
$$\Upsilon = \{\tau \in X(k+1) \;|\; \exists \sigma,\sigma' \subset \tau \mbox{ s.t. } \sigma,\sigma' \in A_i \mbox{ and } \sigma \cap \sigma' \in X(i-1)\setminus A_{i-1} \}.$$

In the following proposition we show how we use this machinery of fat faces. The idea is that either we get a lot of expansion from a certain dimension or we can move down one dimension lower.

\begin{proposition}\label{pro:coboundaries-of-links}
Let $X$ be a $d$-dimensional simplicial complex, $R$ a ring and $\eta>0$ a fatness constant. For any $f \in C^k(X;R)$, $0 \le k \le d-1$, and $0 \le i \le k$,
\begin{align*}
\norm{\delta(f)} \ge{} &\min_{\sigma \in A_i}\left\{\frac{\norm{\delta((f\down\sigma)_{\sigma})}}{\norm{(f\down\sigma)_{\sigma}}}\right\} \Pr[r_k \in A\down r_i \wedge r_i \in A_i] -{} \\[5pt] &(k+1-i)(i+1)\Pr[r_k \in A\down r_{i-1} \wedge r_{i-1} \in A_{i-1}] - \norm{\Upsilon}.
\end{align*}
\end{proposition}
\begin{proof}
By lemma~\ref{lem:when-link-sees-everything-coboundaries-are-global} we know that every local coboundary of $(f\down\sigma)_\sigma$ is also a global coboundary, i.e., $\tau \in \supp(\delta((f\down\sigma)_\sigma)) \;\Rightarrow\; \sigma\cup\tau \in \supp(\delta(f\down\sigma))$. Thus,
\begin{equation}\label{eq:coboundaries-of-links-1}
\norm{\delta((f\down\sigma)_\sigma)} \le
\norm{(\delta(f\down\sigma))_\sigma} =
\Pr[r_{k+1} \in \supp(\delta(f\down\sigma)) \;|\; r_i = \sigma].
\end{equation}

Consider a face $\tau \in \supp(\delta(f\down\sigma))$. By definition, it contains at least one $k$-face $\tau^* \subset \tau$, such that $\tau^* \in A\down\sigma$.  We claim that one of the following cases must occur:
\begin{enumerate}
	\item $\tau$ is a bad face.
	\item $\sigma$ contains a fat $(i-1)$-face $\sigma^* \in A_{i-1}$, such that $\tau^* \in A\down\sigma^*$.
	\item $\tau \in \supp(\delta(f))$.
\end{enumerate}

If $\tau$ is a bad face, the claim holds, so assume that $\tau$ is not a bad face. By definition, there exists a sequence of fat faces $\tau_{k-1} \in A_{k-1}, \tau_{k-2} \in A_{k-2},\dotsc, \tau_{i+1} \in A_{i+1}$, such that $\tau \supset \tau^* \supset \tau_{k-1} \supset \dotsb \supset \tau_{i+1} \supset \sigma$. Let us denote $\tau = \{v_0,v_1,\dotsc,v_{k+1}\}$, $\tau^* = \tau \setminus \{v_{k+1}\}$, $\tau_{k-1} = \tau^* \setminus \{v_k\}$, and so on down to $\sigma = \tau_{i+1} \setminus \{v_{i+1}\}$. Now, if $\tau \setminus \{v_j\} \in A$ for some $j \in \{0,\dotsc, i\}$, then $\tau^* \setminus \{v_j\} \in A_{k-1}$ since it is the intersection of two fat $k$-faces, and then $\tau_{k-1} \setminus \{v_j\} \in A_{k-2}$, and so on down to $\sigma^* = \sigma \setminus \{v_j\} \in A_{i-1}$, and case 2 holds. Otherwise, for any $j \in \{i+1,\dotsc,k\}$, a similar argument shows that if $\tau \setminus \{v_j\} \in A$ then $\tau \setminus \{v_j\} \in A\down\sigma$. It follows that $f$ and $f\down\sigma$ agree on all $k$-faces that are contained in $\tau$, and case 3 holds. Thus,
\begin{equation}\label{eq:coboundaries-of-links-3}
\tau \in \supp(\delta(f\down\sigma)) \quad\Rightarrow\quad
(\tau \in \Upsilon) \vee (\tau^* \in A\down\sigma^*) \vee (\tau \in \supp(\delta(f))).
\end{equation}

Using~\eqref{eq:coboundaries-of-links-3} and summing over all $\tau \in \supp(\delta(f\down\sigma))$ yields
\begin{equation}\label{eq:coboundaries-of-links-4}
\begin{aligned}
\Pr[r_{k+1} \in \supp(\delta(f\down\sigma)) \;|\; &r_i = \sigma] \le{}\\[5pt]
&\Pr[r_{k+1} \in \Upsilon \;|\; r_i = \sigma] +{} \\[5pt]
&(k+1-i)(i+1)\Pr[r_k \in A\down r_{i-1} \wedge r_{i-1} = A_{i-1} \;|\; r_i = \sigma] +{} \\[5pt]
&\Pr[r_{k+1} \in \supp(\delta(f)) \;|\; r_i = \sigma],
\end{aligned}
\end{equation}
where the $(k+1-i)(i+1)$ factor is due the probability that $r_k = \tau^*$ and $r_{i-1} = \sigma^*$ given that $r_{k+1} = \tau \supset \tau^*$ and $r_i = \sigma \supset \sigma^*$. Substituting~\eqref{eq:coboundaries-of-links-1} in~\eqref{eq:coboundaries-of-links-4}, and multiplying and dividing by $\norm{(f\down\sigma)_\sigma} = \Pr[r_k \in A\down\sigma \;|\; r_i = \sigma]$ yields
\begin{equation}\label{eq:coboundaries-of-links-6}
\begin{aligned}
\frac{\norm{\delta((f\down\sigma)_\sigma)}}{\norm{(f\down\sigma)_\sigma}}\Pr[r_k \in f\down\sigma \;|\; &r_i = \sigma] \le{}\\
&\Pr[r_{k+1} \in \Upsilon \;|\; r_i = \sigma] +{} \\[6pt]
&(k+1-i)(i+1)\Pr[r_k \in A\down r_{i-1} \wedge r_{i-1} = A_{i-1} \;|\; r_i = \sigma] +{} \\[6pt]
&\Pr[r_{k+1} \in \supp(\delta(f)) \;|\; r_i = \sigma].
\end{aligned}
\end{equation}

Multiplying~\eqref{eq:coboundaries-of-links-6} by $\Pr[r_i = \sigma]$, summing over all $\sigma \in A_i$, and applying the law of total probability to the right-hand side yields
\begin{align*}
\sum_{\sigma \in A_i}\frac{\norm{\delta((f\down\sigma)_\sigma)}}{\norm{(f\down\sigma)_\sigma}}\Pr[r_k \in f\down\sigma \wedge{} &r_i = \sigma] \le{}\\[-7pt]
&\Pr[r_{k+1} \in \Upsilon] +{} \\[6pt]
&(k+1-i)(i+1)\Pr[r_k \in A\down r_{i-1} \wedge r_{i-1} = A_{i-1}] +{} \\[6pt]
&\Pr[r_{k+1} \in \supp(\delta(f))].
\end{align*}

Taking the minimum over all $\sigma \in A_i$ and rearranging completes the proof.
\end{proof}

It is left to bound the size of the bad faces. The following proposition shows that the size of the bad faces is controlled by the skeleton expansion of the links.

\begin{proposition}[Skeleton expansion implies small set of bad faces]\label{pro:skeleton-expansion-implies-small-set-of-bad-faces}
Let $X$ be a $d$-dimensional simplicial complex, $\eta>0$ a fatness constant and $0 < \alpha \le \eta^{2^{d-1}}$. If for any $\sigma \in X$, the link $X_\sigma$ is an $\alpha$-skeleton expander, then for any subset of $k$-faces $A \subseteq X(k)$, $0 \le k \le d-1$,
$$\norm{\Upsilon} \le \eta(k+1)(k+2)2^{k+2}\norm{A}.$$
\end{proposition}
\begin{proof}
By definition, any bad face $\tau \in \Upsilon$ contains at least one pair of faces $\sigma,\sigma' \subset \tau$ such that $\sigma,\sigma' \in A_i$, $\sigma \cup \sigma' \in X(i+1)$, and $\sigma \cap \sigma' \in X(i-1) \setminus A_{i-1}$ for some $0 \le i \le k$. For any $\tau \in \Upsilon$, choose one such pair $\sigma,\sigma' \subset \tau$ and denote by $\widehat{\tau} = \sigma \cup \sigma'$ and by $\widecheck{\tau} = \sigma \cap \sigma'$. Note that $\widehat{\tau}$ is seen in the link of $\widecheck{\tau}$ as an edge between two fat vertices. Denote by $\Upsilon_i = \{\tau \in \Upsilon \;|\; \widehat{\tau} \in X(i) \}$, so the set of bad faces can be decomposed to $\Upsilon = \bigsqcup_{i=1}^{k+1}\Upsilon_i$. Now,

\begin{align*}
\Pr[r_{k+1} \in \Upsilon] &=
\sum_{i=1}^{k+1}\sum_{\tau \in \Upsilon_i}\Pr[r_{k+1} = \tau] =
\sum_{i=1}^{k+1}\sum_{\tau \in \Upsilon_i}\frac{\Pr[r_{k+1} = \tau \wedge r_i = \widehat{\tau} \wedge r_{i-2} = \widecheck{\tau}]}{\Pr[r_i = \widehat{\tau} \wedge r_{i-2} = \widecheck{\tau} \;|\; r_{k+1} = \tau]} \\&\le
\sum_{i=1}^{k+1}\sum_{\tau \in \Upsilon_i}\binom{k+2}{i+1}\binom{i+1}{i-1}\Pr[r_i = \widehat{\tau} \wedge r_{i-2} = \widecheck{\tau}] \\&\le
\sum_{i=1}^{k+1}\sum_{\tau \in \Upsilon_i}\binom{k+2}{i+1}\binom{i+1}{i-1}(\eta^{2^{k+1-i}} + \alpha)\Pr[r_{i-1} \in A_{i-1} \wedge r_{i-2} = \widecheck{\tau}] \\&\le
\sum_{i=1}^{k+1}\binom{k+2}{i+1}\binom{i+1}{i-1}\frac{2\eta^{2^{k+1-i}}}{\eta^{2^{k+1-i}-1}}\Pr[r_k \in A] \\&\le
(k+2)(k+1)\eta\Pr[r_k \in A]\sum_{i=1}^{k+1}\binom{k+2}{i+1},
\end{align*}
where the second inequality follows by the $\alpha$-skeleton expansion of the links, and the third inequality follows by the law of total probability and by lemma~\ref{lem:upper-bound-on-fat-faces}.
\end{proof}

We can now prove theorem~\ref{thm:existence-of-good-dimension}.

\begin{proof}[Proof of theorem~\ref{thm:existence-of-good-dimension}]
Define the fatness constant $\eta = \alpha^{2^{-d}}$. Now, let $f \in C^k(X;R)$ be a locally minimal $k$-cochain with $\norm{f} \le \alpha \le \eta^{2^{k+1}}$. By lemma~\ref{lem:upper-bound-on-fat-faces} it follows that
$$\norm{A_{-1}} \le \eta^{1-2^{k+1}}\norm{f} \le \eta < 1.$$
But since $X(-1)$ contains only one face, i.e., $\norm{A_{-1}} \in \{0,1\}$, then $\norm{A_{-1}} = 0$. In other words, the empty-set is not a fat face, thus $\Pr[r_k \in A\down r_{-1} \wedge r_{-1} \in A_{-1}] = 0$. Also note that $\Pr[r_k \in A\down r_k \wedge r_k \in A_k] = \norm{f} \ge c_k\norm{f}$.

Now, if $\Pr[r_k \in A\down r_i \wedge r_i \in A_i] \ge c_i\norm{f}$ for all $0 \le i \le k$, then applying proposition~\ref{pro:coboundaries-of-links} on $i=0$ yields
\begin{equation}\label{eq:existence-of-good-dimension-1}
\norm{\delta(f)} \ge \min_{\sigma \in A_0}\left\{\frac{\norm{\delta((f\down \sigma)_\sigma)}}{\norm{(f\down\sigma)_\sigma}}\right\}c_0 - \norm{\Upsilon}.
\end{equation}
Otherwise, let $0 \le j \le k-1$ be the maximal for which $\Pr[r_k \in A\down r_j \wedge r_j \in A_j] < c_j\norm{f}$. Applying proposition~\ref{pro:coboundaries-of-links} on $i = j + 1$ yields
\begin{equation}\label{eq:existence-of-good-dimension-2}
\norm{\delta(f)} \ge \min_{\sigma \in A_i}\left\{\frac{\norm{\delta((f\down \sigma)_\sigma)}}{\norm{(f\down\sigma)_\sigma}}\right\}c_i - (k+1-i)(i+1)c_{i-1} - \norm{\Upsilon}.
\end{equation}

Since $f$ is locally minimal, by lemma~\ref{lem:minimal-cochains-are-closed-under-inclusion}, $(f\down\sigma)_\sigma$ is minimal in $X_\sigma$ for any $\emptyset \ne \sigma \in X$. Thus,
\begin{equation}\label{eq:existence-of-good-dimension-3}
\norm{(f\down\sigma)_\sigma} = \dist((f\down\sigma)_\sigma, B^{k-|\sigma|}(X_\sigma;R)).
\end{equation}

By proposition~\ref{pro:skeleton-expansion-implies-small-set-of-bad-faces} we know that
\begin{equation}\label{eq:existence-of-good-dimension-4}
\norm{\Upsilon} \le \alpha^{2^{-d}}(k+1)(k+2)2^{k+2}\norm{f}.
\end{equation}

Substituting~\eqref{eq:existence-of-good-dimension-3} and~\eqref{eq:existence-of-good-dimension-4} in~\eqref{eq:existence-of-good-dimension-1} or~\eqref{eq:existence-of-good-dimension-2} completes the proof.
\end{proof}

\section{Spherical buildings}\label{sec:spherical-building}

Spherical buildings are very symmetrical complexes with a nice geometrical structure. An example for a spherical building is the following complex. Let $d \in \mathbb{N}$ and $q$ a prime power. Denote by $V = \F_q^d$ the $d$-dimensional vector space over $\F_q$. The vertices of the complex are proper subspaces of $V$ (i.e., not $\{0\}$ and $V$), and its faces are flags of subspaces. The resulting complex is a $(d-2)$-dimensional spherical building (since maximal flags have $d-1$ vertices). For $d=3$ this is the famous "lines versus planes" graph which is known to be an excellent expander.

Any $d$-dimensional spherical building $X$ comes with a collection of $d$-dimensional subcomplexes, called \emph{apartments}, such that all the apartments are isomorphic to each other and for any two faces in the complex there exists an apartment containing both of them. An important fact is that the size of each apartment is bounded by a constant $\theta_d$ which depends only on $d$ (and not on the number of vertices). Also, there exists a group of automorphisms $G \le Aut(X)$ which acts transitively on $X$, i.e., for any $\sigma,\sigma' \in X(k)$, $0 \le k \le d$, there exists $g \in G$ such that $g\sigma = \sigma'$.

In~\cite{EK16}, the authors showed that the spherical building is an $\alpha$-skeleton expander for $\alpha>0$ as small as we want (it is controlled by a parameter called the \emph{thickness} of the building). In~\cite{LMM16}, the authors showed that the spherical building is a coboundary expander, but only over $\F_2$. This is not enough for us as we need coboundary expansion over $\Z$. We follow their strategy and with some modifications we prove the following theorem.
\begin{theorem}\label{thm:the-spherical-building-is-coboundary-expander}
The $d$-dimensional spherical building is a $\beta$-coboundary expander over any ring for $$\beta = \big(2^d\theta_d\big)^{-1}.$$
\end{theorem}

The proof of theorem~\ref{thm:the-spherical-building-is-coboundary-expander} is essentially composed of two propositions. We use its geometrical structure in order to relate the coboundary of a cochain to its distance from the coboundaries. By this relation we over-count each face in the coboundary many times. Then we use the symmetrical structure of the building in order to bound these over-counts.

For any $-1 \le k \le d-1$, we denote by $\mathcal{F}_k = X(d) \times X(k)$ the set of all pairs of top faces and $k$-dimensional faces . For any $(\sigma, \tau) \in \mathcal{F}_k$, let $A_{\sigma,\tau}$ be the complex obtained by the intersection of all the apartments in $X$ which contain both $\sigma$ and $\tau$. Note that if $\tau \subset \tau'$, then $A_{\sigma,\tau} \subset A_{\sigma,\tau'}$.

The following proposition is implied by the geometrical structure of the spherical building. Each apartment of the spherical building is a sphere and any piece of it, as we defined above, is either a sphere or it is contractible. It allows us to relate the coboundary of a cochain to its distance from the coboundaries by over-counting each face as the amount of apartments containing it.

\begin{proposition}\label{pro:spherical-building-homological}
Let $X$ be a $d$-dimensional spherical building and $R$ a ring. For any $f \in C^k(X;R)$, $-1 \le k \le d-1$, and $\sigma \in X(d)$,
$$\dist(f, B^k(X;R)) \le
\sum_{\tau \in X(k)}\norm{\tau} \cdot |\supp(\delta(f)) \cap A_{\sigma,\tau}|.$$
\end{proposition}

The next proposition is implied by the symmetrical structure of the spherical building. Since it possess so many symmetries, all the apartments are spread around it evenly. This implies that any face cannot be contained in many apartments, so we can bound the number of times we over-count each face.

\begin{proposition}\label{pro:spherical-building-symmetrical}
Let $X$ be a $d$-dimensional spherical building. For any $-1 \le k \le d-1$ and $\rho \in X$,
$$\sum_{\substack{(\sigma,\tau) \in \mathcal{F}_k:\\\rho \in A_{\sigma,\tau}}} \norm{\tau} \le \theta_d \cdot |\{\sigma \in X(d) \;|\; \rho \subseteq \sigma\}|.$$
\end{proposition}

We show first how theorem~\ref{thm:the-spherical-building-is-coboundary-expander} is implied by the above two propositions.

\begin{proof}[Proof of theorem~\ref{thm:the-spherical-building-is-coboundary-expander}]
Let $f \in C^k(X;R)$, $-1 \le k \le d-1$. By proposition~\ref{pro:spherical-building-homological},

\begin{equation}\label{eq:the-spherical-building-is-coboundary-expander-1}
\begin{aligned}
|X(d)|\cdot\dist(f,B^k(X;R)) &=
\sum_{\sigma \in X(d)} \dist(f,B^k(X;R)) \\&\le
\sum_{\sigma \in X(d)} \sum_{\tau \in X(k)} \norm{\tau} \cdot |\supp(\delta(f)) \cap A_{\sigma,\tau}| \\&=
\sum_{(\sigma,\tau) \in \mathcal{F}_k} \norm{\tau} \cdot |\supp(\delta(f)) \cap A_{\sigma,\tau}| \\&=
\sum_{\rho \in \text{supp}(\delta(f))} \sum_{\substack{(\sigma,\tau) \in \mathcal{F}_k:\\\rho \in A_{\sigma,\tau}}} \norm{\tau}.
\end{aligned}
\end{equation}

By proposition~\ref{pro:spherical-building-symmetrical},
\begin{equation}\label{eq:the-spherical-building-is-coboundary-expander-2}
\sum_{\substack{(\sigma,\tau) \in \mathcal{F}_k:\\\rho \in A_{\sigma,\tau}}} \norm{\tau} \le
\theta_d \cdot |\{\sigma \in X(d) \;|\; \rho \subseteq \sigma\}| =
\theta_d |X(d)|\binom{d+1}{k+2} \norm{\rho}
\end{equation}

Combining~\eqref{eq:the-spherical-building-is-coboundary-expander-1} and~\eqref{eq:the-spherical-building-is-coboundary-expander-2} yields
$$
\dist(f,B^k(X;R)) \le \sum_{\rho \in \text{supp}(\delta(f))} \theta_d \binom{d+1}{k+2} \norm{\rho} = \theta_d \binom{d+1}{k+2} \norm{\delta(f)},
$$
where rearranging completes the proof.
\end{proof}

\subsection{Proof of proposition~\ref{pro:spherical-building-homological}}

We recall some basic definitions of simplicial complexes. Let $X$ be a $d$-dimensional simplicial complex and $R$ a ring. For any $-1 \le k \le d$, a \emph{$k$-chain} is a linear combination of the $k$-dimensional faces with coefficients in $R$. Denote the space of $k$-chains by
$$C_k(X;R) = \left\{\sum_{\sigma \in X(k)}a_\sigma \!\cdot\! \sigma \;|\; \forall \sigma, a_\sigma \in R\right\}.$$

We fix some arbitrary orientations of the faces in $X$, so when considering a face, there is one fixed ordering of its vertices. Then, the boundary of a $k$-face $(v_0,v_1,\dotsc,v_k) \in X(k)$ is
$$\partial((v_0,\dotsc,v_k)) = \sum_{i=0}^{k} (-1)^i(v_0,\dotsc,v_{i-1},v_{i+1},\dotsc,v_k),$$
and the boundary of a $k$-chain $c \in C_k(X;R)$ is
$$\partial(c) = \sum_{\sigma \in X(k)} a_\sigma \!\cdot\! \partial(\sigma).$$

For ease of notation, for $\sigma = (v_0,\dotsc,v_k)$, we denote $\sigma_i = (v_0,\dotsc,v_{i-1},v_{i+1},\dotsc,v_k)$. Note that the boundary operator commutes with the coboundary operator defined in~\secref{sec:preliminaries}, i.e., for any $k$-cochain $f \in C^k(X;R)$ and a $(k+1)$-face $\sigma \in X(k+1)$,
$$\delta(f)(\sigma) = \sum_{i=0}^k (-1)^i f(\sigma_i) = f(\partial(\sigma)).$$

The following lemma from~\cite{LMM16} shows a nice filling property of the complexes $A_{\sigma,\tau}$ defined above.

\begin{lemma}\label{lem:filling-propoerty-of-apartments}\cite[Claim~3.5]{LMM16}
Let $X$ be a $d$-dimensional spherical building and $R$ a ring. For any $(\sigma,\tau) \in \mathcal{F}_k$, $-1 \le k \le d-1$, and an $i$-chain $c \in C_i(A_{\sigma,\tau};R)$, $0 \le i \le d-1$, if $\partial(c) = 0$ then there exists an $(i+1)$-chain $c' \in C_{i+1}(A_{\sigma,\tau};R)$ such that $\partial(c') = c$.
\end{lemma}

We use this filling property in order to define a family of chains such that each two consecutive chains are related by the boundary operator.

\begin{lemma}
Let $X$ be a $d$-dimensional spherical building and $R$ a ring. There exists a family of chains $$\mathcal{C} = \{c_{\sigma,\tau} \in C_{k+1}(A_{\sigma,\tau};R) \;|\; -1 \le k \le d-1,\; (\sigma,\tau) \in \mathcal{F}_k\},$$ such that $$\partial(c_{\sigma,\tau}) = (-1)^{k+1}\tau + \sum_{i=0}^k(-1)^ic_{\sigma,\tau_i}.$$
\end{lemma}
\begin{proof}
We define $\mathcal{C}$ inductively. For $k=-1$, we have only the empty set $\emptyset \in X(-1)$. For any $\sigma \in X(d)$, choose an arbitrary vertex $v_\sigma \in A_{\sigma,\emptyset}(0)$ and define $c_{\sigma,\emptyset} = v_\sigma$. Then it holds that
$$\partial(c_{\sigma,\emptyset}) = \partial(v_\sigma) = (-1)^0\emptyset,$$
as required. Assume now that $\mathcal{C}$ is defined for any $-1 \le i \le k-1$. For any $(\sigma,\tau) \in \mathcal{F}_k$ define $c_{\sigma,\tau}$ as follows. Consider the $k$-chain $c = (-1)^{k+1}\tau + \sum_{i=0}^{k}(-1)^ic_{\sigma,\tau_i}$, and note that $\partial(c) = 0$ since
\begin{align*}
\partial(c) &= (-1)^{k+1}\partial(\tau) + \sum_{i=0}^k(-1)^i \partial(c_{\sigma,\tau_i}) \\&=
(-1)^{k+1}\partial(\tau) + \sum_{i=0}^k(-1)^i \left((-1)^k\tau_i + \sum_{j=0}^{k-1}(-1)^jc_{\sigma,\tau_{ij}} \right) \\[7pt]&=
(-1)^{k+1}\partial(\tau) + (-1)^k\partial(\tau) + \sum_{i>j}(-1)^{i+j}c_{\sigma,\tau_{ij}} + \sum_{i<j}(-1)^{i+j-1}c_{\sigma,\tau_{ij}} = 0.
\end{align*}

By lemma~\ref{lem:filling-propoerty-of-apartments} it follows that there exists an $(k+1)$-chain $c' \in C_{k+1}(A_{\sigma,\tau};R)$ such that $\partial(c') = c$, so define $c_{\sigma,\tau} = c'$.
\end{proof}

For any $\sigma \in X(d)$ and $0 \le k \le d$, we define the contraction operator
$\iota_\sigma = \iota_{\sigma,k}:C^k(X;R) \to C^{k-1}(X;R)$ as follows. For any $f \in C^k(X)$ and $\tau \in X(k-1)$,
$$\iota_\sigma(f)(\tau) = (-1)^kf(c_{\sigma,\tau}).$$

This contraction operator allows us to relate the coboundary of a cochain to its distance from the coboundaries, as shown in the next lemma.

\begin{lemma}\label{lem:coboundary-plus-contraction}
Let $X$ be a $d$-dimensional spherical building and $R$ a ring. For any $f \in C^k(X;R)$, $0 \le k \le d-1$, and $\sigma \in X(d)$,
$$\delta(\iota_\sigma(f)) + \iota_\sigma(\delta(f)) = f.$$
\end{lemma}
\begin{proof}
For any $\tau \in X(k)$,
\begin{align*}
\delta(\iota_\sigma(f))&(\tau) + \iota_\sigma(\delta(f))(\tau) \\&=
\sum_{i=0}^k(-1)^i(\iota_\sigma(f))(\tau_i) + (-1)^{k+1}(\delta(f))(c_{\sigma,\tau}) \\&=
\sum_{i=0}^k(-1)^i(-1)^k f(c_{\sigma,\tau_i}) + (-1)^{k+1}f(\partial (c_{\sigma,\tau})) \\&=
(-1)^k\sum_{i=0}^k(-1)^i f(c_{\sigma,\tau_i}) + (-1)^{k+1}\left((-1)^{k+1}f(\tau) + \sum_{i=0}^k(-1)^if(c_{\sigma,\tau_i})\right) =
f(\tau).
\end{align*}
\end{proof}

We can now prove proposition~\ref{pro:spherical-building-homological}.
\begin{proof}[Proof of proposition~\ref{pro:spherical-building-homological}]
Let $f \in C^k(X;R)$, $-1 \le k \le d-1$. By lemma~\ref{lem:coboundary-plus-contraction}, for any $\sigma \in X(d)$,
\begin{equation}\label{eq:spherical-building-homological-1}
\norm{\iota_\sigma(\delta(f))} = \norm{f - \delta(\iota_\sigma(f))} \ge \dist(f, B^k(X;R)).
\end{equation}

Note that for any $\tau \in X(k)$,
\begin{align*}
\iota_\sigma(\delta(f))(\tau) \ne 0 \quad&\Rightarrow\quad
\delta(f)(c_{\sigma,\tau}) \ne 0 \\[5pt]&\Rightarrow\quad
\exists \rho \in \supp(\delta(f)) \cap \supp(c_{\sigma,\tau}) \\[5pt]&\Rightarrow\quad
\exists \rho \in \supp(\delta(f)) \cap A_{\sigma,\tau},
\end{align*}
which yields that
\begin{equation}\label{eq:spherical-building-homological-2}
\norm{\iota_\sigma(\delta(f))} =
\sum_{\tau \in \text{supp}(\iota_\sigma(\delta(f)))}\norm{\tau} \le
\sum_{\tau \in X(k)} \norm{\tau} \cdot |\supp(\delta(f)) \cap A_{\sigma,\tau}|.
\end{equation}

Combining~\eqref{eq:spherical-building-homological-1} and~\eqref{eq:spherical-building-homological-2} finishes the proof.
\end{proof}

\subsection{Proof of proposition~\ref{pro:spherical-building-symmetrical}}

The key point of the proof is that the spherical building possess so many symmetries, so for any face, only a small portion of the apartments contains it. We show it formally in the following lemma.

\begin{lemma}\label{lem:spherical-building-many-symmetries}
Let $X$ be a $d$-dimensional spherical building and $G \le Aut(X)$ the group that acts transitively on $X$. For any $\rho \in X$ and $(\sigma,\tau) \in \mathcal{F}_k$, $-1 \le k \le d-1$,
$$\frac{|\{g \in G \;|\; g\rho \in A_{\sigma,\tau}\}|}{|G|} \le
\theta_d \frac{|\{\sigma \in X(d) \;|\; \rho \subseteq \sigma\}|}{|X(d)|}.$$
\end{lemma}
\begin{proof}
For any $\rho \in X$, denote by $G_\rho = \{g \in G \;|\; g\rho = \rho\}$ the stabilizer of $\rho$, and consider the quotient $G/ G_\rho$. The elements in $G/ G_\rho$ are equivalence classes of the form $gG_\rho = \{gh \;|\; h \in G_\rho\}$. Consider a face $\sigma \in X(d)$ such that $\rho \subseteq \sigma$. The elements in $G_\rho$ can move $\sigma$ only to other $d$-dimensional faces which contain $\rho$. It follows that for any equivalence class $gG_\rho$, the number of $d$-dimensional faces that the elements in $gG_\rho$ can move $\sigma$ is bounded by the number of $d$-dimensional faces which contain $\sigma$. Since $G$ is transitive, there must be enough equivalence classes to cover all $X(d)$. Thus,

\begin{equation}\label{eq:spherical-building-many-symmetries-1}
\frac{|G|}{|G_\rho|} =
\left|\faktor{G}{G_\rho}\right| \ge
\frac{|X(d)|}{|\{\sigma \in X(d) \;|\; \rho \subseteq \sigma\}|},
\end{equation}
where the equality follows by Lagrange's theorem. Next, note that for any $\rho' \in X$, there are $|G_\rho|$ elements $g \in G$ for which $g\rho = \rho'$. Therefore, for any $(\sigma,\tau) \in \mathcal{F}_k$,
$$
|\{g \in G \;|\; g\rho \in A_{\sigma,\tau}\}| \le |A_{\sigma,\tau}|\cdot|G_\rho| \le
\theta_d \frac{|\{\sigma \in X(d) \;|\; \rho \subseteq \sigma\}|}{|X(d)|} |G|,
$$
where the second inequality follows by~\eqref{eq:spherical-building-many-symmetries-1}.
\end{proof}

We can now prove proposition~\ref{pro:spherical-building-symmetrical}.

\begin{proof}[Proof of proposition~\ref{pro:spherical-building-symmetrical}]
Note that for any $\rho \in X$ and $g \in G$,
\begin{equation}
\sum_{\substack{(\sigma,\tau) \in \mathcal{F}_k:\\\rho \in A_{\sigma,\tau}}}\norm{\tau} =
\sum_{\substack{(g\sigma,g\tau) \in \mathcal{F}_k:\\g\rho \in gA_{\sigma,\tau}}}\norm{\tau} =
\sum_{\substack{(g\sigma,g\tau) \in \mathcal{F}_k:\\g\rho \in A_{g\sigma,g\tau}}}\norm{\tau} =
\sum_{\substack{(\sigma,\tau) \in \mathcal{F}_k:\\g\rho \in A_{\sigma,\tau}}}\norm{\tau}.
\end{equation}
Thus, it is possible to change the order of summation, i.e.,
$$
\sum_{g \in G} \sum_{\substack{(\sigma,\tau) \in \mathcal{F}_k:\\\rho \in A_{\sigma,\tau}}} \norm{\tau} =
\sum_{(\sigma,\tau) \in \mathcal{F}_k} \sum_{\substack{g \in G:\\g\rho \in A_{\sigma,\tau}}} \norm{\tau}.
$$

It follows that for any $\rho \in X$,
\begin{align*}
\sum_{\substack{(\sigma,\tau) \in \mathcal{F}_k:\\\rho \in A_{\sigma,\tau}}} \norm{\tau} &=
\frac{1}{|G|} \sum_{g \in G} \sum_{\substack{(\sigma,\tau) \in \mathcal{F}_k:\\\rho \in A_{\sigma,\tau}}} \norm{\tau} =
\frac{1}{|G|} \sum_{(\sigma,\tau) \in \mathcal{F}_k} \sum_{\substack{g \in G:\\g\rho \in A_{\sigma,\tau}}} \norm{\tau} \\&=
\sum_{(\sigma,\tau) \in \mathcal{F}_k} \norm{\tau} \frac{|\{g \in G \;|\; g\rho \in A_{\sigma,\tau}\}|}{|G|} \\&\le
\theta_d \frac{|\{\sigma \in X(d) \;|\; \rho \subseteq \sigma\}|}{|X(d)|} \sum_{(\sigma,\tau) \in \mathcal{F}_k} \norm{\tau} \\[5pt]&=
\theta_d |\{\sigma \in X(d) \;|\; \rho \subseteq \sigma\}|,
\end{align*}
where the inequality follows by lemma~\ref{lem:spherical-building-many-symmetries}.
\end{proof}

\bibliographystyle{alpha}
\bibliography{Bibliography}

\begin{thebibliography}{LMSS01}

\bibitem[CS13]{CS13}
J.~H. Conway and N.~J.~A. Sloane.
\newblock {\em Sphere packings, lattices and groups}, volume 290.
\newblock Springer Science \& Business Media, 2013.

\bibitem[EK16]{EK16}
S.~Evra and T.~Kaufman.
\newblock Bounded degree cosystolic expanders of every dimension.
\newblock In {\em Proceedings of the 48th Annual {ACM} {SIGACT} Symposium on
  Theory of Computing, {STOC} 2016, Cambridge, MA, USA, June 18-21, 2016},
  pages 36--48, 2016.

\bibitem[Gal63]{Gal63}
R.~G. Gallager.
\newblock Low-density parity-check codes.
\newblock {\em MIT Press, Cambridge, MA}, 1963.

\bibitem[Gro10]{Gro10}
M.~Gromov.
\newblock {Singularities, Expanders and Topology of Maps. Part 2: from
  Combinatorics to Topology Via Algebraic Isoperimetry}.
\newblock {\em Geometric And Functional Analysis}, 20(2):416--526, 2010.

\bibitem[KKL14]{KKL14}
T.~Kaufman, D.~Kazhdan, and A.~Lubotzky.
\newblock {Ramanujan Complexes and Bounded Degree Topological Expanders}.
\newblock In {\em Foundations of Computer Science (FOCS), 2014 IEEE 55th Annual
  Symposium on}, pages 484--493, 2014.

\bibitem[LM06]{LM06}
N.~Linial and R.~Meshulam.
\newblock Homological connectivity of random 2-complexes.
\newblock {\em Combinatorica}, 26(4):475--487, 2006.

\bibitem[LMM16]{LMM16}
A.~Lubotzky, R.~Meshulam, and S.~Mozes.
\newblock Expansion of building-like complexes.
\newblock {\em Groups, Geometry, and Dynamics}, 10(1):155--175, 2016.

\bibitem[LMSS01]{LMSS01}
M.~G. Luby, M.~Mitzenmacher, M.~A. Shokrollahi, and D.~A. Spielman.
\newblock Improved low-density parity-check codes using irregular graphs.
\newblock {\em IEEE Transactions on information Theory}, 47(2):585--598, 2001.

\bibitem[LPS88]{LPS88}
A.~Lubotzky, R.~Phillips, and P.~Sarnak.
\newblock Ramanujan graphs.
\newblock {\em Combinatorica}, 8(3):261--277, 1988.

\bibitem[LSV05a]{LSV05.2}
A.~Lubotzky, B.~Samuels, and U.~Vishne.
\newblock Explicit constructions of ramanujan complexes of type
  $\widetilde{A}_d$.
\newblock {\em European Journal of Combinatorics}, 26(6):965--993, 2005.

\bibitem[LSV05b]{LSV05.1}
A.~Lubotzky, B.~Samuels, and U.~Vishne.
\newblock Ramanujan complexes of type $\widetilde{A}_{d}$.
\newblock {\em Israel Journal of Mathematics}, 149(1):267--299, 2005.

\bibitem[Lub]{Lub}
A.~Lubotzky.
\newblock Personal communication.

\bibitem[Lub14]{Lub14}
A.~Lubotzky.
\newblock Ramanujan complexes and high dimensional expanders.
\newblock {\em Japanese Journal of Mathematics}, 9(2):137--169, 2014.

\bibitem[Pin73]{Pin73}
M.~S. Pinsker.
\newblock On the complexity of a concentrator.
\newblock In {\em 7th International Telegraffic Conference}, volume~4, pages
  1--318, 1973.

\bibitem[Spi96]{Spi96}
D.~A. Spielman.
\newblock Linear-time encodable and decodable error-correcting codes.
\newblock {\em IEEE Transactions on Information Theory}, 42(6):1723--1731,
  1996.

\bibitem[SS96]{SS96}
M.~Sipser and D.~A. Spielman.
\newblock Expander codes.
\newblock {\em IEEE Transactions on Information Theory}, 42(6):1710--1722,
  1996.

\end{thebibliography}

\end{document}